\documentclass[11pt]{amsart}
\usepackage{amssymb,multicol,xcolor,hyperref,here}
\usepackage{graphicx,pstool,psfrag}
\usepackage[english]{babel}
\definecolor{MyDarkBlue}{rgb}{0,0.29,0.7}
\hypersetup{
    colorlinks=true,       
    linkcolor=MyDarkBlue,          
    linkbordercolor=MyDarkBlue,
    citecolor=MyDarkBlue,
    citebordercolor=MyDarkBlue,        
    filecolor=MyDarkBlue,      
    urlcolor=MyDarkBlue           
}

\newtheoremstyle{plain}
  {10pt}
  {10pt}
  {\it}
  {0pt}
  {\bf}
  {}
  {\newline}
  {}
\newtheoremstyle{definition}
  {10pt}
  {10pt}
  {}
  {0pt}
  {\bf}
  {}
  {\newline}
  {}
\theoremstyle{plain}

\newtheorem{theorem}{THEOREM}[section]

\newtheorem*{theorem*}{Theorem}
\newtheorem*{lemma*}{Lemma}
\theoremstyle{definition}

\newtheorem{example}[theorem]{\textit{Example}}
\newtheorem{application}[theorem]{\textit{Application}}
\newtheorem{remark}[theorem]{\textit{Remark}}

\numberwithin{equation}{section}
\title[On Absolute and Relative Change]{On Absolute and Relative Change}
\author[S.~Brauen, P. Erpf \& M.~Wasem]{Silvan Brauen, Philipp Erpf and Micha Wasem}
\date{\today}
\begin{document}\maketitle
\normalsize
\begin{abstract}
Based on an axiomatic approach we propose two related novel one-para\-meter families of indicators of change which put in a relation classical indicators of change such as absolute change, relative change and the log-ratio.
\end{abstract}
\vspace{1cm}

\section{Introduction}
One of the most basic concept in statistics is change (take e.g.\ the change of a quantity in time). Usually, change is either expressed in absolute or relative terms. As we will show below, already the interpretation of such most basic indicators might be difficult: For a fixed choice of a measurement unit, absolute and relative change might be limited in comparing changes at different scales (see Example \ref{example1} below) and it could be desirable to find a suitable indicator of change which combines the information obtained from absolute and relative change. Based on a relaxation of an axiomatic characterization of relative change introduced by \cite[p. 43]{vartia} (see Section \ref{lambdasection}), we will introduce a one-parameter family $f_\lambda$ of indicators of change which is particularly well-behaved with respect to a change of measurement units and which includes absolute and relative change as special cases.\\

In another direction it is revealed that relative change has certain conceptual deficiencies like its lack of additivity and antisymmetry (see \cite{vartia}, \cite{wetherell} and \cite{schurer}), which can be remedied by replacing it with a log-ratio which seems to appear for the first time in \cite{mcalister} (see also \cite{aitchison1} and \cite{aitchison2}). We will introduce a new antisymmetric and additive one-parameter family of indicators $F_\lambda$ which interpolates between absolute change and the log-ratio and we will show how $f_\lambda$ and $F_\lambda$ are interlinked thus providing a general framework for absolute change, relative change and the log-ratio.\\

As applications we show how $f_\lambda$ and $F_\lambda$ allow to meaningfully compare changes at different scales (see Example \ref{example1}) and how they uncover a simple relationship between marginal functions and economic elasticity (see Example \ref{elasticity}).\\

Throughout the article, will denote the set of (strictly) positive real numbers by $\mathbb R_+$, absolute change by $\operatorname{abs}(x,y):=y-x$ and relative change by $\operatorname {rel}(x,y):=(y-x)/x$, ($x\ne 0$).
 
\section{Exhibition of the Problem}
We will now illustrate possible limitations of absolute and relative change when analyzed alone.
\begin{example}\label{example1}
Take the example case of a company selling a good through five different sales channels I to V. The amount of units sold at a past time (past value) is denoted by $x$ and the one at present time (present value) by $y$. Suppose the company measures the following numbers:
\begin{center}
\footnotesize
\begin{tabular}{|l|c|c|c|c|} \hline
  Channel & Past Value $x$ & Present Value $y$ & $\mathrm{abs}(x,y)$ & $\mathrm{rel}(x,y)$\\
I&10&20 & 10 & 100\%\\
II&500 & 570 & 70 & 14\%\\
III&140 & 210 & 70 & 50\%\\
IV&35 & 70 & 35 & 100\%\\
V& 80 & 135 & 55 & 68.75\%\\
\hline \end{tabular}\end{center}
\normalsize
\begin{enumerate}
\item Considering absolute change, it is impossible to distinguish between II and III although III is better than II in relative terms. In absolute terms, the channels I and IV have rather low growth while their relative growth is comparably large in contrast to II and III with a high absolute growth even if their relative growth is comparably small. 
\item Considering relative change, I and IV perform equally well, although IV has a higher absolute change.
\item Neither absolute nor relative change allow to meaningfully compare changes at different scales (i.e.\ is I better than II? Or: which of the sales channels is best?) How could one compare channel V to the other channels?
\end{enumerate}

The indicator $f_\lambda$ we construct below depends on a number $\lambda\in\mathbb R$ and -- as we will show below -- is an interpolation of absolute and relative change provided $\lambda\in[0,1]$. In the present example -- since we would like to combine information from absolute and relative change -- we will choose $\lambda=\frac12$ (see Remark \ref{properties} (2)). In this case the indicator is given by $f_{1/2}:\mathbb R_+^2\to \mathbb R$, $
f_{1/2}(x,y)=(y-x)/\sqrt{x}.$
The values of $f_{1/2}$ for all sales channels are listed below:

\begin{center}
\footnotesize
\begin{tabular}{|l|c|c|c|} \hline
  Sales Channel & Past Value $x$ & Present Value $y$ & $f_{1/2}(x,y)$\\
I&10&20 & 3.16 \\
II&500 & 570 & 3.13 \\
III&140 & 210 & 5.92\\
IV&35 & 70 & 5.92 \\
V & 80 & 135 & 6.15\\
\hline \end{tabular}\end{center}
\normalsize

The upshot is that all the values become comparable using a single number (III and IV are equally good) and V is rated best without being absolutely nor relatively best.
\end{example}

\section{Construction of $f_\lambda$}\label{lambdasection}
In \cite[p. 43]{vartia}, an indicator of relative change is characterized by a function $r:\mathbb R^2_+\to\mathbb R$ satisfying
\begin{enumerate}
\item $r(x,y) = 0 \Longleftrightarrow x = y$,
\item $r(x,y) \lessgtr 0 \Longleftrightarrow x \lessgtr y$,
\item $r$ is continuous and $y\mapsto r(x,y)$ is increasing.
\item $r(Cx,Cy)=r(x,y)$ for all $C>0$.
\end{enumerate}
Since we seek for an indicator of change that interpolates between absolute and relative change, we will relax axiom (4) since it excludes $r(x,y)=\mathrm{abs}(x,y)$. Our idea is to replace the scaling invariance (4) by a relaxed \emph{relative scaling invariance} (see below). Furthermore, since $y\mapsto \mathrm{abs}(x,y)$ and $y\mapsto \mathrm{rel}(x,y)$ are affine linear, we include affine linearity as an axiom too. As we will show in the sequel, our set of axioms will determine a family of indicators which is unique up to a positive multiplicative constant.
\subsection{Axioms}
\begin{enumerate}
\item \textbf{Affine Linearity in the Second Argument.} A map $f:\mathbb R_+^2\to\mathbb R$ is said to be \emph{affine linear} in the second argument if $f(x,y)=m(x)y+b(x),$ where $m$ and $b$ are values which may still depend on $x$.
\item \textbf{Naturality.} We call $f$ natural, if $f$ is continuous and moreover $f(x,y)\lessgtr 0$ if $x\gtrless y$ and $f(x,x)=0$ for all $x>0$ (i.e.\ $f$ assigns a positive (negative) number to growth (decrease) and zero to stagnation).
\item \textbf{Relative Scaling Invariance. }$f$ is relative scaling invariant, if it behaves relatively invariant under a change of scale, this is, for all $C>0$ and all pairs $(x,y),(\bar x,\bar y)\in\mathbb R_+^2$ the equation
\begin{equation}\label{scale}
f(x,y) \cdot f(C\bar x,C\bar y)=f(\bar x,\bar y)\cdot f(Cx,Cy) 
\end{equation}
should hold true. This means that the measurement units in which $x$ and $y$ are measured do not matter in a relative sense.\end{enumerate}

\subsection{Construction} We have the following theorem:
\begin{theorem}
If a map $f:\mathbb R^2_+\to \mathbb R$ satisfies the axioms (1)-(3), then it holds that
$$
f(x,y)=\mathcal C\cdot\frac{y-x}{x^\lambda}$$
for some $\mathcal C>0$ and $\lambda\in\mathbb R$.
\end{theorem}
\begin{proof}
The affine linearity condition together with the zero assignment for stagnation implies that $f(x,x)=m(x)x+b(x)=0,$ which forces $b(x)=-m(x)x$. Hence, the general form of $f$ will be $f(x,y)=m(x)(y-x).$
Since \eqref{scale} holds for all $C>0$ and all admissible pairs $(x,y)$ and $(\bar x,\bar y)$, it follows that
\begin{equation}\label{scaling2}
f(Cx,Cy)=g(C)f(x,y)\end{equation}
for some continuous function $g:\mathbb R_+\to\mathbb R$. Hence (see \cite[p. 163]{efthimiou}) $g(C)$ is a solution to the \emph{Power Cauchy Equation} which has the only continuous solutions $
g(C)=C^\mu,~\mu\in\mathbb R$. Replacing this solution in equation \eqref{scaling2}, we deduce that $
m(Cx)= C^{\mu-1}m(x)$, which implies that $m$ is a homogeneous function of degree $\mu-1$. According to \emph{Euler's Homogeneous Function Theorem}, $m$ satisfies the ordinary differential equation
$
m'(x)-\frac{\mu-1}{x}m(x)=0
$
with the solution $m(x)=\mathcal C x^{-\lambda}$ for some constant $\mathcal C$ and $\lambda := 1-\mu$. Putting everything together we obtain
$$
f(x,y)=\mathcal C\cdot\frac{y-x}{x^\lambda}.
$$
and the sign of the constant, $\mathcal C>0$, is determined by the naturality assumption. This finishes the proof.
\end{proof}
We will henceforth set $\mathcal C=1$ (see Remark \ref{properties} (2) below) and define
$$
f_\lambda(x,y):=\frac{y-x}{x^\lambda}.
$$

\begin{remark}\label{properties}
We now list a few properties of $f_\lambda$:
\begin{enumerate}
\item\textbf{Formal Structure of a Cobb-Douglas Function. }Observe that we might write $
f_\lambda (x,y) = \mathrm{rel}(x,y)^\lambda \cdot \mathrm{abs}(x,y)^{1-\lambda}$ and therefore interpret $f_\lambda$ formally as a Cobb-Douglas function
$Y(L,K)=A L^\beta K^\alpha$
with constant returns to scale, where $A=1$, $L=\mathrm{rel}(x,y)$, $K=\mathrm{abs}(x,y)$, $\alpha = 1-\lambda$ and $\beta=\lambda$.
\item \textbf{Generalizing Absolute/Relative Change and Calibration. }The normalization choice $\mathcal C=1$ is justified by the observation that
$$f_0(x,y)=\mathrm{abs}(x,y)\text{ and }f_{1}(x,y)=\mathrm{rel}(x,y).$$
This indicates in how far $f_\lambda$ is a generalization of absolute and relative change. The number $\lambda$ allows to intentionally weigh between the two and therefore serves as a calibration. If one assigns the same value to two pairs $(x,y)$ and $(\bar x,\bar y)$ such that $x\ne y$, $\bar x\ne \bar y$ and $x\ne\bar x$, the value of $\lambda$ can be determined using the formula$$
\lambda = \frac{\ln\left(\frac{\bar y - \bar x}{y-x}\right)}{\ln\left(\frac{\bar x}{x}\right)}.
$$
A natural choice for being in the middle between absolute and relative is $\lambda=\frac12$, the value chosen in Example \ref{example1}.
\item \textbf{Relative Scaling Invariance.} If the unit in which $x$ and $y$ are measured is $\mathrm u$, then the unit of $f_\lambda(x,y)$ is $\mathrm u^{1-\lambda}$. In this way $f_\lambda$ cannot directly be interpreted but it useful in comparing different pairs $(x,y)$ and $(\bar x,\bar y)$ since the unit-free quotient $
{f_{\lambda}(\bar x,\bar y)}/{f_{\lambda}(x,y)}$ will be independent of the choice of the unit $\mathrm u$ according to the relative scaling invariance \eqref{scale}. 
\item \textbf{From Differences to Quantities.} Using the structure of the Cobb-Douglas function above and replacing the differences $\mathrm{abs}(x,y)$ and $\mathrm{rel}(x,y)$ by (absolute and relative) quantities, one obtains a generalization of absolute and relative quantities: For an absolute quantity $y\geqslant 0$ relative to another quantity $x>0$ one would obtain the function
$$
(x,y)\mapsto\left(\frac{y}{x}\right)^\lambda y^{1-\lambda} = \frac{y}{x^\lambda},
$$
which recovers the absolute quantity provided $\lambda = 0$ and the relative one if $\lambda = 1$. This function is linear in the second argument and it satisfies the relative scaling invariance (axiom (3) above).
\end{enumerate}
\end{remark}



\section{An Antisymmetric, Additive Variant of $f_\lambda$}
Common criticism on relative change includes its failure of \emph{antisymmetry}, i.e.\ generally $\mathrm{rel}(x,y)\ne - \mathrm{rel}(y,x)$ and its failure of \emph{additivity} i.e.\ generally $$\mathrm{rel}(x,y)+\mathrm{rel}(y,z)\ne\mathrm{rel}(x,z)$$ (see \cite{vartia}, \cite{wetherell} and \cite{schurer}). 
It is observed in \cite{vartia}, that the log-ratio $r(x,y)=\ln(y/x)$ is the \emph{unique} indicator of relative change that is antisymmetric, additive and \emph{normed}, i.e.\ the linearization of $y\mapsto r(x,y)$ around $x$ is given by $\mathrm{rel}(x,y)$. It is therefore natural to ask if an indicator $F_\lambda:\mathbb R_+^2\to\mathbb R$ exists which is antisymmetric, additive and normed in the sense that the linearization of $y\mapsto F_\lambda(x,y)$ around $x$ is given by $f_\lambda(x,y)$. The following theorem answers this question affirmatively:
\begin{theorem}\label{Flambda}
The indicator
$$F_\lambda(x,y)=\begin{cases}\displaystyle \frac{y^{1-\lambda}-x^{1-\lambda}}{1-\lambda},&\text{if }\lambda \ne 1\\
\hfill \ln(y/x),&\text{if }\lambda = 1\end{cases}
$$
is the unique antisymmetric and additive indicator such that the linearization of $y\mapsto F_\lambda(x,y)$ around $x$ is given by $f_\lambda(x,y)$.
\end{theorem}
\begin{proof}
We start by deducing a certain structure of $F_\lambda$ from differentiating the additivity property
$
F_\lambda(x_1,x_2)+F_\lambda(x_2,t) = F_\lambda(x_1,t)\quad \forall x_1,x_2,t\in \mathbb R_+
$
with respect to $t$. Hence $
\partial_yF_\lambda(x_2,t)=\partial_yF_\lambda(x_1,t)$ for all $x_1,x_2,t\in\mathbb R_+$ which implies that $\partial_yF_\lambda$ only depends on $t$ and in particular it holds that 
\begin{equation}\label{independence}
\partial_yF_\lambda(x,t) = \partial_yF_\lambda(t,t)
\end{equation}
for all $(x,t)\in\mathbb R_+^2$. It follows that $F_\lambda(x,y)$ is \emph{additively separable} (it decompoases into a sum of two single variable functions depending on $x$ and $y$ respectively). The linearization of $y\mapsto F_\lambda(x,y)$ around $x$ should equal $f_\lambda$, so we obtain
$$
F_\lambda(x,x)+\partial_yF_\lambda(x,x)(y-x)=\frac{y-x}{x^\lambda}
$$
and since $F_\lambda(x,x)=0$ by antisymmetry of $F_\lambda$, it follows that
\begin{equation}\label{derivative}
\partial_yF_\lambda(x,x) = \frac{1}{x^\lambda}
\end{equation}
and hence
$$\begin{aligned}F_\lambda(x,y) & =\int_{x}^y \partial_yF_\lambda(x,t)\,\mathrm dt \stackrel{\eqref{independence}}{=} \int_{x}^y \partial_yF_\lambda(t,t)\,\mathrm dt \stackrel{\eqref{derivative}}{=} \int_{x}^y \frac{1}{t^\lambda}\,\mathrm dt\\
& =\begin{cases}\displaystyle \frac{y^{1-\lambda}-x^{1-\lambda}}{1-\lambda},&\text{if }\lambda \ne 1\\
\hfill \ln(y/x),&\text{if }\lambda = 1.\end{cases}
\end{aligned}$$
It is readily checked that $F_\lambda(x,y)+F_\lambda(y,z)=F_\lambda(x,z)~\forall x,y,z\in\mathbb R_+$.
\end{proof}

\begin{remark} We now list a few properties of $F_\lambda$:
\begin{enumerate}
\item Observe that $F_\lambda$ puts absolute change ($\lambda = 0$) and the log-ratio ($\lambda = 1$) into a natural relation since
$$
\lim_{\lambda \to 1}\frac{y^{1-\lambda}-x^{1-\lambda}}{1-\lambda}=\ln(y/x).
$$
Furthermore, $F_\lambda$ inherits the naturality and the relative scaling invariance from $f_\lambda$ (but not the affine linearity in the second argument).
\item Using a Taylor expansion of $y\mapsto F_\lambda(x,y)$ around $x$ we obtain\\

$\displaystyle
F_\lambda(x,y)=f_{\lambda}(x,y)+\sum_{k=2}^n\frac{(-1)^{k+1}\Gamma(\lambda+k-1)}{k!x^{k+\lambda-1}\Gamma(\lambda)}(y-x)^k+\mathcal O((x-y)^{n+1}),
$\\

\noindent
where $\Gamma$ denotes Euler's Gamma function and we note that the truncated series also satisfies the relative scaling invariance and the naturality property. Truncation after the linear term and and using the Lagrange form of the remainder, we obtain a quadratic global bound on the difference between $F_\lambda$ and $f_\lambda$.
$$|F_\lambda(x,y)-f_\lambda(x,y)|\leqslant\frac{\lambda\cdot (y-x)^2}{\min\{x,y\}^{1+\lambda}}.
$$
\item If $x=1$ (which corresponds to a choice of measurement units), then the function $y\mapsto F_\lambda(1,y)$ equals a Box-Cox transformation of parameter $1-\lambda$ (see \cite{boxcox}).\\
In this case we have $F_0(1,y)=y, F_{\frac{1}{5}}(1,y)=\tfrac{5}{4}(\sqrt[5]{y^4}-1), F_{\frac12}(1,y)=2(\sqrt y-1)$ and $F_1(1,y)=\ln y$ (see Figure \ref{graph}).\end{enumerate}
\begin{figure}[H]
\includegraphics[scale=0.8]{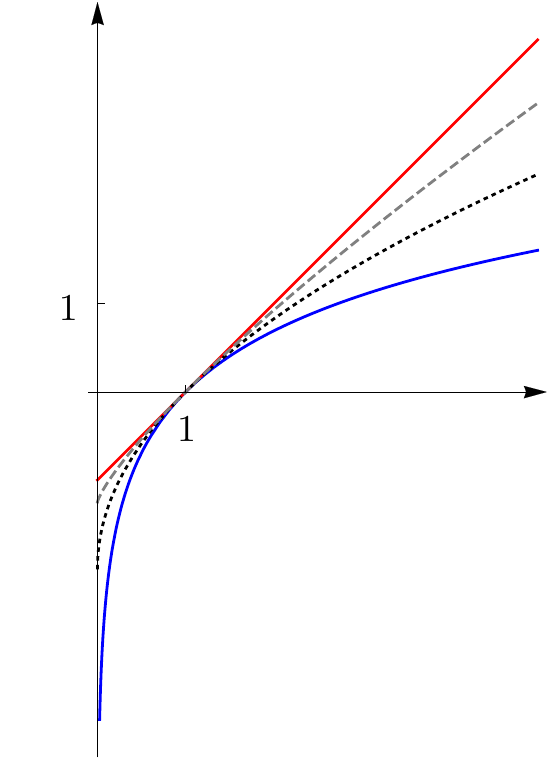}
\caption{The plot shows the graphs of $F_0(1,y)$, $F_1(1,y)$ (solid), $F_{\frac12}(1,y)$ (dotted) and  $F_{\frac15}(1,y)$ (dashed) for $y\in(0,5]$.}\label{graph}
\end{figure}
\end{remark}

\section{Application}
Every quantity which involves either absolute or relative changes can naturally be generalized using $f_\lambda$ (or $F_\lambda$):

\begin{application}[Relation between Marginal Functions and Elasticity]\label{elasticity}
The indicators $f_\lambda$ and $F_\lambda$ uncover a relationship between marginal functions and elasticity: For a differentiable economic function $g(x)$, the associated marginal function is given by its derivative $g'(x)$. The corresponding elasticity function is defined by
$$
\varepsilon_g(x)=\lim_{y\to x}\frac{(g(y)-g(x))/g(x)}{(y-x)/x}= g'(x)\cdot\frac{x}{g(x)},
$$
which is the limit of a quotient of relative changes. Replacing the relative changes in this definition by $f_\lambda$, we obtain a generalized elasticity function
$$
\varepsilon_{g}^\lambda(x)=\lim_{y\to x}{\frac{(g(y)-g(x))/g(x)^\lambda}{(y-x)/x^\lambda}}=g'(x)\cdot\left(\frac{x}{g(x)}\right)^\lambda
$$
which puts the marginal function ($\lambda=0$) and the classical elasticity ($\lambda=1$) into a natural relation.\\

\end{application}



\section{The Choice of $\lambda$}\label{choice}
The choice of $\lambda$ depends heavily on the context. If one wishes to interpolate ``symmetrically'' between absolute and relative change, the following conceptual evidence reveals $\lambda = \frac12$ as a good choice: In order to fix $\lambda$ we require that a scaling of the absolute change by a factor $C>0$ (for a fixed relative change) amounts to the same as the same scaling of the relative change (for a fixed absolute change). Observe that the pairs $(x,y)$ and $(Cx,Cy)$ satisfy $\operatorname{abs}(Cx,Cy)=C\operatorname{abs}(x,y)$ and $\operatorname{rel}(Cx,Cy)=\operatorname{rel}(x,y)$, hence the absolute change scales with $C$ whereas the relative change is preserved. The pairs $(x,y)$ and $\left(\frac{x}{C},y-x+\frac{x}{C}\right)$ satisfy

$$\begin{aligned}\operatorname{abs}\left(\frac{x}{C},y-x+\frac{x}{C}\right) & =\operatorname{abs}(x,y)\\
\operatorname{rel}\left(\frac{x}{C},y-x+\frac{x}{C}\right) & =C\operatorname{rel}(x,y),\end{aligned}$$
hence the relative change scales with $C$ whereas the absolute change is preserved. The parameter $\lambda$ will now be chosen such that $f_\lambda$ satisfies
$$
f_\lambda(Cx,Cy)=f_\lambda\left(\frac{x}{C},y-x+\frac{x}{C}\right).
$$
This equation reduces to $C^{1-\lambda} = C^{\lambda}$ and hence $\lambda=\frac{1}{2}$. We will add a concrete example to illustrate the choice. As a reference pair we choose $(1,2)$ which satisfies $f_\lambda(1,2) = 1$ for every choice of $\lambda$. The pair $(2,4)$ has an absolute change which is twice as big as the one of $(1,2)$ but $\operatorname{rel}(1,2)=\operatorname{rel}(2,4)$. The unique pair which has the same absolute change as $(1,2)$ but twice its relative change is given by $\left(\frac12,\frac32\right)$. Then $f_\lambda(2,4)=2^{1-\lambda}$, $f_\lambda\left(\frac12,\frac32\right)=2^\lambda$ and $2^{1-\lambda}=2^\lambda$ iff $\lambda=\frac12$. Hence $\lambda$ is chosen in such a way that doubling either the absolute or the relative change amounts to the same result.\\

In the perspective of $f$ as a Cobb-Douglas function we note that the choice $\lambda=\frac{1}{2}$ has the following consequence: The marginal rate of substitution at a point $(L,K)$ of a Cobb-Douglas function $Y(L,K)=A L^\beta K^\alpha$ is given by
\begin{equation}\label{marginalrate}
\frac{\partial_L Y(L,K)}{\partial_K Y(L,K)}=\frac{\beta L^{\beta-1}K^\alpha}{\alpha L^\beta  K^{\alpha-1}}=\frac{\beta}{\alpha}\cdot \frac{K}{L}.
\end{equation}
In our case, whenever $L=\mathrm{rel}(x,y)$, $K=\mathrm{abs}(x,y)$, $\alpha = 1-\lambda$ and $\beta=\lambda$, the latter expression in \eqref{marginalrate} equals $\lambda/(1-\lambda)\cdot x$. This quantity equals the past value $x$ exactly when $\lambda = \frac12$.

\section{Conclusion}
We have shown that the indicator of change $f_\lambda$ can be singled out from some simple and natural axioms. In this way, the basic concepts of absolute and relative change are identified as special cases of a more general quantity. Building upon this indicator, a new antisymmetric and additive indicator $F_\lambda$ is constructed, which relates absolute change to the log-ratio. Our analysis therefore contains the main result of \cite{vartia} as a special case.\\

Moreover, we have obtained a novel generalized elasticity function which uncovers a relationship between two classical concepts in economics -- marginal functions and elasticity.\\


\subsection*{Acknowledgements}
The authors would like to thank Thomas Mettler for pointing out the formal structure of $f_\lambda$ as a Cobb-Douglas function.

\begin{flushright}

\textit{Rivella AG,\\ Neue Industriestrasse 10,
CH-4852 Rothrist}\\ \url{silvan.brauen@rivella.ch}
\\[0.2cm]

\textit{Institute for Research on Management of Associations,\\
Foundations and Co-op\-eratives (VMI),\\ University of Fribourg,\\ Boulevard de P\'erolles 90, CH-1700 Fribourg} \\\url{philipp.erpf@unifr.ch}\\[0.2cm]

\textit{HEIA Fribourg,\\ HES-SO University of Applied Sciences and Arts Western Switzerland,\\ P\'erolles 80, CH-1700 Fribourg}\\ \url{micha.wasem@hefr.ch}

\end{flushright}
\end{document}